\documentclass[letterpaper,11pt]{article}
\usepackage{amsmath, amsthm, amssymb}
\usepackage[usenames, dvipsnames]{color}
\usepackage[normalem]{ulem} 
\usepackage{fullpage}
\usepackage[noend]{algorithmic}
\usepackage{tikz, pgfplots}
\usetikzlibrary{shapes.gates.logic.US,trees,positioning,arrows}
\usepackage{enumerate}
\usepackage{thmtools} 
\usepackage{multirow}
\usepackage{multicol}
\usepackage{comment}
\usepackage{verbatim}
\usepackage{bbm}
\usepackage{hyperref}
\usepackage{chngpage}
\usepackage{xspace,color}
\usepackage{graphicx}
\usepackage{caption}
\usepackage{makecell}
\usepackage{subcaption}
\usepackage{enumitem,linegoal}
\setlist{nosep}
\usepackage[linesnumbered,ruled,vlined]{algorithm2e}
\usepackage{comment}	
\usepackage{ctable}					
\newtheorem{claim}{Claim}

\usepackage{mathtools}
\usepackage[flushleft]{threeparttable}
\usepackage{verbatim}
\usepackage{float}

\usepackage{footnote}
\makesavenoteenv{tabular}
\makesavenoteenv{table}
\usepackage{relsize}
\usepackage{thm-restate}

\usepackage{xspace}

\usepackage[margin=1in]{geometry}

\newcolumntype{?}[1]{!{\vrule width #1}}

\definecolor{crimsonglory}{rgb}{0,0,0}

\definecolor{commentcolor}{rgb}{0.75, 0.0, 0.2}

 \newtheorem{lemma}{Lemma}

 \newtheorem{definition}{Definition}

\newif\ifqed
 

\usetikzlibrary{arrows,shapes,snakes,automata,backgrounds,petri,calc}

\newtheoremstyle{named}{}{}{\itshape}{}{\bfseries}{.}{.5em}{\thmnote{#3}}
\theoremstyle{named}

\newcounter{proccnt}

\newcommand{\konote}[1]{}

\def\F{\textsf{E}}

\def\Z{\mathcal{Z}}
\def\F{\mathcal{F}}

\newcommand{\eps}{\varepsilon}

\title{Drawing Competitive Districts in Redistricting}

\author{
	 Gabriel Chuang\thanks{Computer Science, Columbia University, New York NY 10027, USA, gtc2117@columbia.edu. Supported by an NSF Graduate Fellowship.} \and Oussama Hanguir\thanks{Lyft, Inc.  Work done while a student at Columbia University. Email: oh2204@columbia.edu.} \and Clifford Stein\thanks{Industrial Engineering and Operations Research, Columbia University, New York, NY 10027, USA. Email: cliff@ieor.columbia.edu.  Research supported in part  by NSF grant  CCF-2218677 and ONR grant ONR-13533312, and by the Wai T. Chang Chair in Industrial Engineering and Operations Research.}
}

\begin{document}
	\newcommand{\ignore}[1]{}
\sloppy
\date{}

\maketitle


\thispagestyle{empty}

\begin{abstract}

In the process of redistricting, one important metric is the number of \emph{competitive districts}, that is, districts where both parties have a reasonable chance of winning a majority of votes. Competitive districts are important for achieving proportionality, responsiveness, and other desirable qualities; some states even directly list competitiveness in their legally-codified districting requirements. In this work, we discuss the problem of drawing plans with at least a fixed number of competitive districts. In addition to the standard, ``vote-band'' measure of competitivenesss (i.e., how close was the last election?), we propose a measure that explicitly considers ``swing voters'' - the segment of the population that may choose to vote either way, or not vote at all, in a given election. We present two main, contrasting results. First, from a computational complexity perspective, we show that the task of drawing plans with competitive districts is NP-hard, even on very natural instances where the districting task itself is easy (e.g., small rectangular grids of population-balanced cells). Second, however, we show that a simple hill-climbing procedure can in practice find districtings on real states in which \emph{all} the districts are competitive. We present the results of the latter on the precinct-level graphs of the U.S. states of North Carolina and Arizona, and discuss trade-offs between competitiveness and other desirable qualities. 
\end{abstract}

\setcounter{page}{1}
\section{Introduction}\label{sec:intro}

In the United States, 
\emph{redistricting} is the task of geographically dividing a state into a fixed number of regions called \emph{districts}, each of which elects one representative to a legislative body (such as the U.S. House of Representatives or a state legislature). 
The process is prone to various types of manipulation, collectively known as \emph{gerrymandering}, in which parties draw districting maps that are optimized for particular outcomes. For example, a party may wish to maximize seats in which their preferred voters constitute a majority, protect their party's incumbents, or force opposition-party incumbents to run against each other. This process often results in many districts that are uncompetitive, i.e., districts in which one party's voters constitute such a large majority that voters are denied any meaningful choice and the winning party is effectively pre-determined.




To combat gerrymandering, many quantitative and qualitative measures capturing various normative criteria have been proposed. These include notions of proportionality \cite{DS2022_RedistrictingProportionality}, responsiveness \cite{N2019_WhatCriteriaShould}, partisan symmetry \cite{W2016_ThreeTestsPractical}, typicality \cite{DDS2020_ComputationalApproachMeasuring}, stability under perturbation \cite{DNW2020_HomologicalPersistenceGerrymandering}, and many others (e.g., \cite{CPRV2019_DeclinationMetricDetect}, \cite{CRSV2022_GeographyElectionOutcome}; see \cite{W2019_ComparisonPartisanGerrymanderingMeasures} for a comparison of several partisan-based measures). 

In this work, we focus on \emph{competitive districts}: those where one expects elections to be close and highly contested, i.e., where the outcomes of future elections are not pre-determined by the geography of the district. There are three reasons we consider competitiveness to be of particular importance: 

\begin{enumerate}
    \item Several jurisdictions in the United States explicitly require competitiveness as a quality that their districting plans must satisfy. Colorado, for example, has a requirement that plans must ``maximize the number of politically competitive districts,'' where competitive is defined as ``having a reasonable potential for the party affiliation of the district’s representative to change at least once between federal decennial censuses'' \cite{CCDMS2022_ColoradoContextCongressional}. See \cite{DDS2020_ComputationalApproachMeasuring} for an overview. 
    \item Having competitive districts is key for the responsiveness of the plan. Responsiveness is a measure of how much a given change in popular vote translates into a change in the proportion of seats held by a particular party. For example, the 14-district Congresional map enacted in North Carolina in 2023 is highly non-responsive: the state could swing 5 points more Democratic and result in zero more Democratic-held seats, or \emph{17} points more Republican and result in only one Republican pickup. This lack of responsiveness is a direct consequence of there being only one competitive seat of the fourteen. 
    \item Competitive elections are generally seen as promoting various positive civic qualities, such as voter engagement, high turnout, close attention to local issues, and others. 
\end{enumerate}

The number of competitive districts in the country has been shrinking rapidly. In 2020, only 45 of 435 districts have a Cook Partisan Voting Index between R+3 and D+3, down from 107 in 1999 \cite{W2023_RealignmentMoreRedistricting}; Wasserman estimates that around half of the lost swing seats are due to changes to district boundaries, rather than ``true'' changes in electoral behavior (e.g. changing voter preferences, geographic polarization, etc.). 

\subsection{Contributions}

Motivated by the desirable qualities of competitive districts, 
we study a version of the districting problem where the goal is to draw maps with at least some fixed number of competitive districts. In addition to the intuitive notion of competitiveness, where a district is considered competitive if recent elections have been decided by close margins (e.g., by 5\% or less), we consider a characterization that relies on separately counting ``swing voters'' -- that is, voters who have a reasonable chance of voting for either party, a formulation that aligns with the Colorado requirement of having districts with a ``reasonable chance of being won by either party''. We show that the problem of maximizing swing districts is NP-hard for both our characterization and the standard model, even on instances where the underlying districting problem is polynomial-time solvable: that is, the hardness comes from the competitive-districts requirement, rather than the inherent hardness of balanced graph partitioning. Despite this, we show that a simple hill-climbing procedure can achieve a very high number of competitive districts without significantly sacrificing compactness, equal population, and other desirable qualities. We demonstrate the results on data from North Carolina and Arizona, showing that it is possible to make every district competitive (although we do not necessarily advocate for doing so).  We also complement the NP-completeness results by giving restrictions that make the problem of maximizing swing districts more tractable.

\section{Related Work}\label{sec:relatedwork}

\paragraph{Competitiveness of districting plans.}

Among the plethora of proposed evaluation metrics for districting plans, Deford et al. \cite{DDS2020_ComputationalApproachMeasuring} take a comprehensive look at various criteria that aim to operationalize competitiveness, including ``evenness'' (how close is the vote share to 50\%?), ``typicality'' (how close is the vote share to the national or statewide average?), and ``vote-band'' metrics (does the vote fall within a fixed percentage of 50\%, or the statewide average?). They observe that ``there is no guarantee that it is even possible to construct a plan with a large number of ... districts [that fall within a given vote-band] while adhering to reasonable compactness and boundary preservation norms,'' an idea we will expand on in this work. 
We will adopt the vote-band metric, because its binary nature allows us to easily express the problem of drawing competitive districts as a decision problem. In addition, we will propose another metric based on swing voters. 

In the same work, Deford et al. \cite{DDS2020_ComputationalApproachMeasuring} also
conduct an extensive ensemble analysis of how many competitive districts arise in ``typical'' districtings, and present two hill-climbing algorithms for optimizing directly for competitive districts. We conduct similar experiments for both vote-band and swing-voter metrics, but use a randomized weighted scheme that incorporates compactness directly via the isoperimetric score. 

Other works that explore the competitiveness of enacted and proposed plans largely use margins in recent elections as their measurement for competitiveness, including \cite{CCDMS2022_ColoradoContextCongressional}, which investigates plans for Colorado, and \cite{HHG2018_GerrymanderingIncumbencyDoes}, which studies whether independent commissions tend to draw more competitive plans.

\paragraph{Computational complexity of redistricting.}
There has also been work exploring hardness of the districting problem in a computational complexity sense, including work on auditing for local deviating groups \cite{KTAM2022_AllPoliticsLocal}, on minimizing the margin of victory in non-geographically-bounded districts \cite{SCDG2019_MinimizingMarginVictory}, and on describing classes of graphs for which various redistricting tasks are NP-hard \cite{IKKO2021_AlgorithmsGerrymanderingGraphs}. Notably, even the balanced graph partition problem (i.e., drawing population-balanced districts, without any other restrictions) is NP-hard, even for planar graphs \cite{A_BalancedGraphPartitioning}. 
Given this, \cite{KMV2019FairRedistrictingHard} show that the problem of drawing districts where each party wins at least $c$ seats is NP-hard \emph{even on instances where valid (contiguous and population-balanced) maps can be found in polynomial time}, by showing a reduction from Var-Linked Planar 3-SAT. We will adopt this structure for our hardness results: we will show reductions that create instances where population-balanced districts can be drawn in polynomial time but drawing competitive districts (both vote-band competitive and swing) is hard.

\section{Preliminaries and Problem Formulation}\label{sec:preliminaries}
\subsection{Voting setting}
Our setting consists of a set of voters distributed over a given geographic area (such as a U.S. state), represented by a (typically planar) graph $G$ with $n$ cells (nodes). The cells are some fixed geographic units (such as counties, precincts, etc.) and edges represent adjacency\footnote{Depending on the jurisdiction, legal requirements may mandate either rook adjacency or queen adjacency.}.

We assume that there are only two parties in consideration, Party \textbf{A} and Party \textbf{B}; every voter may be either a partisan voter or a swing voter. Specifically, each cell $c_i : i \in \{1, \cdots, n\}$ has the following four quantities, all non-negative integers: 
\begin{itemize}
    \item $Pop_i \geq 0$ indicating the total population of the cell, which we abbreviate $Pop_i$; 
    \item $a_i$, the number of voters that vote for Party \textbf{A}; 
    \item $b_i$, the number of voters that vote for Party \textbf{B}; 
    \item optionally, $s_i$, the number of swing voters (who may vote either way). 
\end{itemize}



\subsection{Districtings and Competitive Districts}
For a fixed $d \in \{2, \ldots, n-1\}$, a $d$-\textit{districting} is a partition of the cells of  $G$ into $d$ disjoint subgraphs $D_1, \cdots, D_d$, which we call \emph{districts}. For any district $D_j$, we have:
$$ Party_A(D_j)= A_j=\sum\limits_{c_i \in D_j} a_i\ , \qquad 
Party_B(D_j) = B_j = \sum\limits_{c_i \in D_j} b_i \ ,\qquad 
Swing(D_j) = S_j = \sum\limits_{c_i \in D_j} s_i \ . $$ 
Since every voter is either A, B, or swing, we have $Pop(D_j) = Pop_j = A_j + B_j + S_j$.

In general, a $d$-districting of a graph $G$ is $\varepsilon$-valid if it satisfies the following constraints: 
\begin{enumerate}
    \item[C1.] \emph{Contiguity}, i.e., the subgraph induced by $D_j$ must be connected for each $j=1, \cdots, d$. Thus,  each district is contiguous, which is typically required by law.  
    \item[C2.] $\varepsilon$-\emph{Population-balance}, i.e. $(1-\eps)\left(\frac{Pop(G)}{d}\right) \leq Pop(D_j) \leq (1+\eps)\left(\frac{Pop(G)}{d}\right)$ for each $j=1, \cdots, d$, which ensures that each district has approximately the same population, which is required under the ``One Person, One Vote'' rule \cite{S2014_DemocracyDoorstepStory}. We only consider $\eps < \frac16$, as $\eps\geq\frac16$ allows one district to have \emph{double} the population of another district. Different jurisdictions may require different values of $\varepsilon$. 
\end{enumerate}

\subsubsection{Competitiveness}

We consider two notions that aim to capture the competitiveness of a given district. 

First, we consider the intuitive notion that a district is competitive if the most recent election was decided by a very close margin - for example, 51\% of votes cast for Party A and 49\% cast for Party B. Deford et al. \cite{DDS2020_ComputationalApproachMeasuring} call this ``vote-band'' competitiveness, named for the ``band'' of outcomes (for example, 45-55\%) that the vote share should fall into in order for the election to be considered competitive\footnote{Deford et al. consider both centering the vote band around 50\% and centering it around the statewide or nationwide average. For simplicity, we only discuss the former.}. This notion does \emph{not} depend on counting swing voters separately, so, in this context, 
$s_i = 0$ for all cells $c_i$. 

\begin{definition}
    A district $D_j$ is \emph{$\delta$-Vote-Band Competitive ($\delta$-VBC)} iff $\frac{A_j}{Pop_j}, \frac{B_j}{Pop_j} \in \left[\frac12 - \delta, \frac12 + \delta\right]$.
\end{definition}

An alternative notion is that a district is competitive if the outcome depends on how the swing voters vote. In particular, elections can be both close and uncompetitive. For example, a district comprised of 40\% Party A voters, 40\% Party B voters, and 20\% swing voters (who may vote either way) is likely to be highly competitive, whereas a district with 52\% Party A voters, 47\% Party B voters, and 1\% swing voters is less likely to be competitive. 

\begin{definition}
    A district $D_j$ is \emph{Swing} if $S_j \geq |A_j - B_j|$.
\end{definition}

This formulation more directly captures the notion that districts are competitive if there is a reasonable chance for either party to win. However, evaluation of this metric depends on having a reasonable estimate of the number of such voters; we discuss using statistical methods of Ecological Inference for this task in Appendix \ref{apdx:EI}. In particular, while polling data is valuable for estimating voter transitions, we need this information at a precinct-by-precinct level to do redistricting analysis, and polls are mostly conducted at a national or statewide level. 

This formulation also allows us to capture the fact that, due to partisan polarization, elections in the US are often determined by turnout, rather than vote-switching \cite{HG2010_EstimatingElectoralEffects}. For example, we can count ``unreliable partisan voters'' (those who will either vote for Party A or stay home) as half of a ``true'' swing voter (those who may vote for Party A or Party B), as their decision will affect the eventual vote margin in their precinct, district, or state by half as much.

\subsection{Maximizing Competitive Districts}

\begin{definition}[$\delta$-VBC-Max]
    Given a graph $G$ with $n$ cells $c_1, \ldots, c_n$, $d \in [2,n-1]$, $\varepsilon > 0$, and $\delta>0$, the \emph{$\delta$-VBC District Maximization Problem} is to compute a $\varepsilon$-valid $d$-districting that maximizes the number of $\delta$-VBC districts.
\end{definition}

\begin{definition}[Swing-Max]
    Given a graph $G$ with $n$ cells $c_1, \ldots, c_n$, $d \in [2,n-1]$ and $\varepsilon>0$, the \emph{Swing District Maximization Problem} is to compute a $\varepsilon$-valid $d$-districting that maximizes the number of swing districts.
\end{definition}
\section{Hardness results}
Our main results in this section are to show that the $\delta$-VBC District Maximization Problem and the Swing District Maximization Problem are both NP-hard. On its own, deciding the existence of any $\varepsilon$-valid $d$-districting for a graph $G$ is NP-hard, regardless of competitive districts. Therefore, we show that both $\delta$-VBC-Max and Swing-Max are hard, even on instances where: 
\begin{itemize}
    \item the underlying graph $G$ is a grid,
    \item the number of districts $d$ is 2, and
    \item there exists a polynomial-time computable $\varepsilon$-valid $d$-districting. 
\end{itemize} 
This setup is similar that of Kueng et al. \cite{KMV2019FairRedistrictingHard}, who show that drawing districtings where both parties have at least $c$ seats is hard even on instances where balanced-population districting is easy. 

\begin{restatable}[Vote-Band Hardness]{thm}{VBCHard}
    For all positive $\eps < \frac16, \delta < \frac12$, the $\delta$-VBC $\eps$-valid District Maximization Problem ($\delta$-VBC-Max($\eps$)) is NP-hard, even on instaces where $G$ is a grid, $d=2$, and $\eps$-valid districtings (that are not $\delta$-VBC-maximal) are poly-time computable.  
    \label{thm:VBC-hard}
\end{restatable}

\begin{restatable}[Swing Hardness]{thm}{SwingMaxHard}
    For all positive $\eps < \frac16$, the Swing $\eps$-valid District Maximization Problem (Swing-Max($\eps$)) is NP-hard, even on instaces where $G$ is a grid, $d=2$, and $\eps$-valid districtings (that are not Swing-maximal) are poly-time computable.\footnote{This constraint of $\eps < \frac16$ is extremely lenient: in practice, most districtings have population balance under $1\%$, and $\eps=\frac16$ would allow one district to have \emph{double} the population of another - which is certainly illegal.}
    \label{thm:SwingMax-hard}
\end{restatable}

We will prove both of these via a reduction from Subset Sum \cite{KT2006_AlgorithmDesign}. Specifically, we will reduce through an intermediate problem: 

\begin{definition}[FPSP($\delta$)]
    For any fixed positive $\delta < \frac12$, the \emph{Fixed-Proportion Subset Problem (FPSP)} takes a list of tuples $S = [(a_1, b_1) \cdots (a_n, b_n)]$, where: 
    \[ \sum_{i=1}^n a_i = \left(\frac12+\delta\right)\sum_{i=1}^n (a_i + b_i) \qquad \text{ and } \qquad a_1 + b_1 = a_2 + b_2 = \cdots = a_n + b_n \]
    The task is to find nonempty proper subsets $S_1, S_2\subsetneq [n]$ that partition $[n]$, such that 
    \[\sum_{i \in S_1} a_i = \left(\frac12+\delta\right) \sum_{i \in S_1} (a_i + b_i) \qquad \text{and} \qquad \sum_{i \in S_2} a_i = \left(\frac12+\delta\right) \sum_{i \in S_2} (a_i + b_i). \footnote{The existence of any $S_1$ satisfying the first condition is sufficient to have $S_2 = [n] \setminus S_1$ satisfy the second condition, but we explicitly write both for clarity.} \]
\end{definition}

\begin{figure}
    \centering 
    \includegraphics*[width=0.9\textwidth]{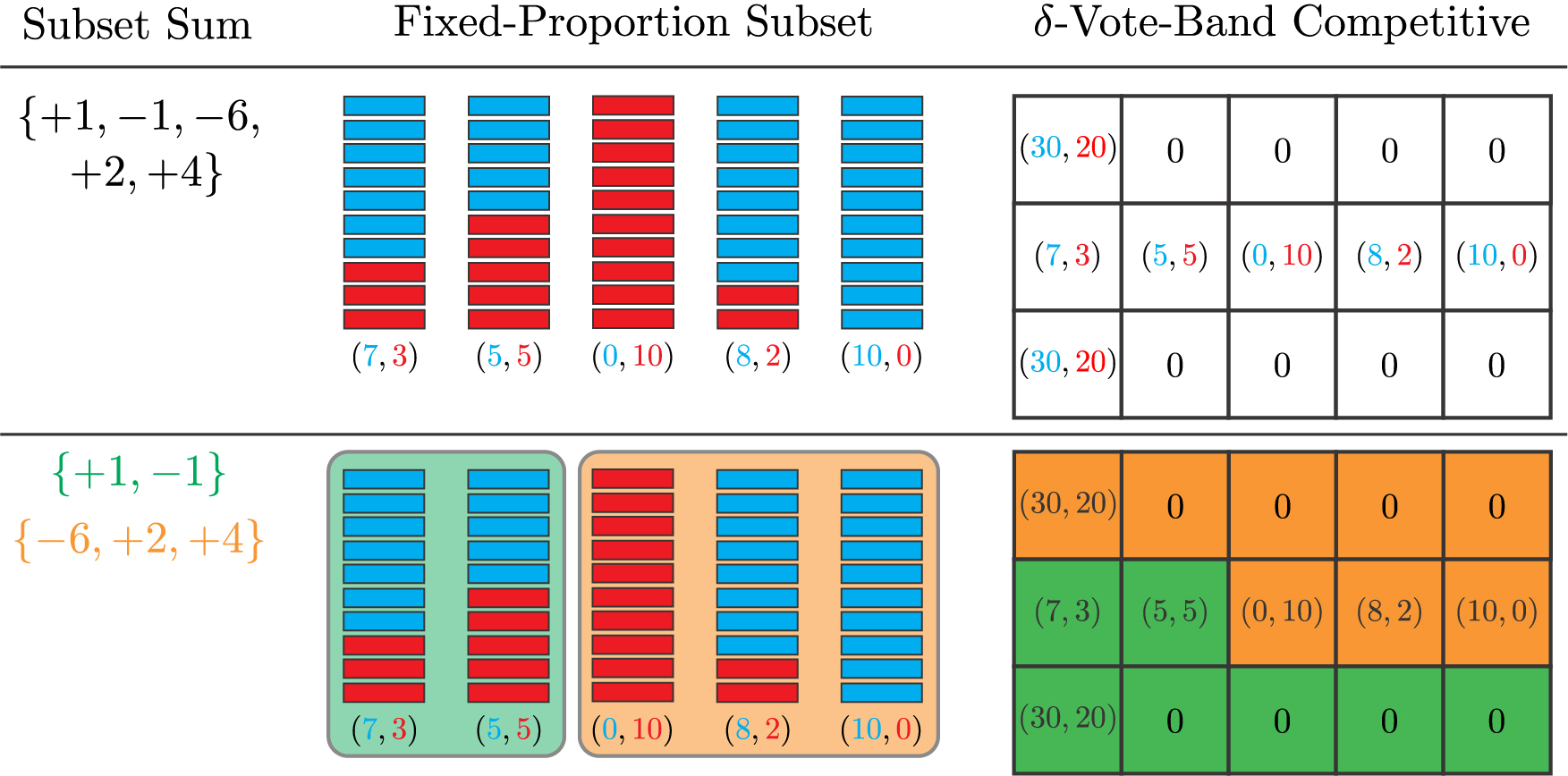}
    \caption{The chain of reductions from Subset Sum to FPSP to $\delta$-VBC-Max, for $\delta = 0.1, \eps=\frac16$. Given the Subset Sum instance $T$ (left), we construct an instance of FPSP (center) where bin $i$ has {\color{Cerulean} $6+t_i$ type-A units} and {\color{Red} $4-t_i$ type-B units}. This induces a districting instance (right), where the central row has {\color{Cerulean} Party A} and {\color{Red} Party B} voters corresponding to type-A and type-B units. The {\color{Green} green} and {\color{Orange} orange} districts (where each district has population in $(\frac12 \pm \eps)pop_{\text{total}}$ and exactly 60\% Party A voters) correspond to a FPSP partition (where each partition has exactly 60\% type-A units), which in turn corresponds to a solution to the original Subset Sum instance.} 
    \label{fig:reduction}
\end{figure}

Intuitively, the Fixed-Proportion Subset Problem asks to partition a population of units with $(\frac12 + \delta)$ type-$a$ units and $(\frac12 - \delta)$ type-$b$ units, which are grouped into equal-size bins, into two subsets, where each subset is also exactly $(\frac12 + \delta)$ type-$a$ and $(\frac12 - \delta)$ type-$b$. 

\begin{restatable}[FPSP Hardness]{thm}{FPSPHard}
    The Fixed-Proportion Subset Problem is NP-hard for all $\delta \in [0, \frac12)$. 
    \label{thm:FPSP-hard}
\end{restatable}

We prove Theorem~\ref{thm:FPSP-hard} in Appendix~\ref{apdx:proofs}, using a reduction from Subset Sum. 

We prove Theorem \ref{thm:VBC-hard} via a reduction from FPSP($\delta$) to $\delta$-VBC-Max.

\begin{proof} Fix $\eps$ and $\delta$. Let $S = [(a_1, b_1) \cdots (a_n, b_n)]$ be an instance of FPSP($\delta$), and let $Z = \sum_{i=1}^n (a_i + b_i)$. Our redistricting instance (``state'') for $\delta$-VBC-Max($\eps$) will be a $3\times n$ grid of cells $C_{i,j}$, where: 
\begin{itemize}
    \item The top left and bottom left cells $C \in \{C_{1,1}, C_{3,1}\}$ will each have
    \[pop(C) = \frac12 P, \qquad Party_A(C) = \left(\frac12 + \delta\right)\frac12 P, \qquad  Party_B(C) = \left(\frac12 - \delta\right)\frac12 P\ ,\]
    where $P= \frac{\frac12 + \eps}{\frac12 - \eps} Z$. 
    \item All other cells in the first and third row have zero population\footnote{One can easily modify the reduction to have all cells have nonzero population by multiplying all other cells by $2(n+2)+1$ and having the first- and third-row cells have $pop(C) = 2, Party_A(C) = 1, Party_B(C) = 1$.}: $pop(C_{1,\cdot}) = pop(C_{3, \cdot}) = 0$.
    \item For each cell $C_{2,i}$ in the second row (for $i=1, 2, \cdots n$), 
    \[pop(C_{2,i}) = (a_i + b_i) \qquad Party_A(C_{2,i}) = a_i  \qquad Party_B(C_{2, i}) = b_i\ . \]
\end{itemize} 

The total population of the state is 
$pop_{total} = Z+P = \frac{\frac12 - \eps + \frac12 + \eps}{\left(\frac12 - \eps\right)}Z = \frac{1}{\frac12 - \eps}Z\ $.

\textbf{Observation 0}: Since all cells $C_{2, 1\cdots n}$ have the same population, it is trivial to draw an $\eps$-valid districting: one can assign the first row and half of the second row to $D_1$, and the rest to $D_2$. That is, one can district this instance in polynomial time, if competitiveness is not considered.

\textbf{Observation 1}: The two left ``corner'' cells $C_{1,1}$ and $C_{3,1}$ must be assigned to different districts. If they were assigned the same district $D_i$ (along with at least one center-row cell $c_j$, required for connectivity), that district's population would exceed $\left(\frac12 + \eps\right)pop_{total}$:
\[ pop(C_{1,1}) + pop(C_{3,1}) + pop(C_{2,j}) = P + (a_j+b_j) > P = \frac{\frac12 + \eps}{\frac12 - \eps}Z =\left(\frac12+\eps\right)pop_{total}.\]

\textbf{Observation 2}: Each district must contain at least one of the center-row cells $C_{2,i}$; a ``corner'' cell alone is under the population bounds. 

Besides these two cells, the structure of the state allows all other second-row cells to be assigned to either district while respecting contiguity. For example, one can assign all cells in the first row to district 1 and all cells in the third row to district 2.

\textbf{Observation 3}: $\frac12+\delta$ of the overall population is Party A voters: 
\[\text{total }Party_A= \sum_{i=1}^n a_i + \left(\frac12+\delta\right)\frac{\frac12+\eps}{\frac12-\eps}Z
= \left(\frac12 + \delta\right)Z + \left(\frac12+\delta\right)\frac{\frac12+\eps}{\frac12-\eps}Z
 = \left(\frac12 + \delta\right)pop_{total}, \]
where the last equality holds because $pop_{total} = \frac{1}{\frac12-\eps}Z$. 

Let $\{D_1, D_2\}$ be a $\epsilon$-valid $\delta$-Vote-Band-Competitive districting on this state. 
We claim that $S_1 = \{i : C_{2, i} \in D_1\}, S_2 = \{i : C_{2, i}\in D_2\}$ is a valid solution to the FPSP($\eps$) instance. 

For $S_1, S_2$ to be a solution to FPSP($\eps$), we must show (a) both are nonempty (which follows from Observation 2 above), and (b) $\sum_{i \in S_1} a_i = \left(\frac12 + \delta\right) \sum_{i \in S_1} (a_i + b_i)$.

Since both districts' margins fall in the $\frac12 \pm \delta$ vote band, each one must have \emph{exactly} $\frac12 + \delta$ Party A voters; if (for example) $D_1$ had $<\frac12 + \delta$ Party A voters, then $D_2$ would end up with $>\frac12 + \delta$ Party A voters, falling outside the vote band. So, 
$$
    \frac12 + \delta  = \frac{Party_A(D_1)}{Party_A(D_1) + Party_B(D_1)} 
    = \frac{\left(\frac12 + \delta\right) \frac12 P + \sum_{i \in S_1} a_i}{\frac12 P + \sum_{i \in S_1} (a_i + b_i)} \ . $$
    Solving this equation for $\sum_{i \in S_1} a_i$, yields
    $ \sum_{i \in S_1} a_i = \left(\frac12 + \delta\right) {\sum_{i \in S_1} (a_i + b_i)} $
as required. Thus, $S_1, S_2$ are a valid solution to the FPSP instance. 
\end{proof}

The chain of reductions, using the Subset Sum instance $T=\{1, -1, -6, 2, 4\}$ as an example, is shown in Fig.~\ref{fig:reduction}. 

The proof for Swing district maximization is nearly identical. We present it in Appendix~\ref{apdx:proof_sw}.

\begin{restatable}[Hardness for more than two districts]{cor}{MoreDistricts}
    For any $k,d$ with $2 \leq k \leq d$, it is NP-hard to generate an $\eps$-valid $d$-districting plan where at least $k$ districts are competitive (for either $\delta$-VBC or SWING). 
    \label{cor:dgeq2}
\end{restatable}

The proof for Corollary \ref{cor:dgeq2} is given in Appendix~\ref{apdx:proof_dgeq2}. 

\textbf{Algorithms for Special Cases.} In Appendix~\ref{apdx:poly}, we present some special cases of graphs which admit polynomial time algorithms. Although these cases are somewhat artificial if expressed as standard redistricting instances, they may be of interest to those interested in districting-flavored problems on constrained structures, such as road networks or low-population grids. 

\subsection{Discussion of Complexity Results} 
First, we re-emphasize that we have shown that the hardness is a result of the competitive district maximization, not the intrinsic hardness of the balanced partition problem, because the instances we construct can be partitioned to population balance $\eps=0$ trivially. 

There are reasons to see this reduction as somewhat more ``natural'' than previous hardness results, which often construct instances with strange features (for example, \cite{KMV2019FairRedistrictingHard} constructs instances with that have long, tendril-like geography with many holes; \cite{SCDG2019_MinimizingMarginVictory} deals with a setting without geographic contiguity requirements; \cite{KTAM2022_AllPoliticsLocal} creates instances where half of the nodes are connected to only one other node (i.e., there are many ``donut'' precincts, where one precinct entirely surrounds another)). In contrast, we construct instances where the underlying graph is a rectangular grid; many states (especially in the Midwest) have precinct graphs that have this approximate structure. 

On the other hand, our reduction relies  heavily on the hard cutoff at the edge of a desired ``vote band'' -- that is, the hardness comes from distinguishing, in a binary way, between a district that is 59.99\% Party A and one that is 60.01\% Party A. In practice, the difference is likely to be negligible. Therefore, one may hope that we can, in practice, draw maps that significantly increase the number of competitive districts. We explain our approach to doing so in the subsequent section. 

\section{Algorithms, Heuristics, and Experiment Setup}
Given the hardness of drawing competitive districting plans (as measured both in vote-band and swing-voter flavors), even on fairly constrained and natural varieties of graphs, one may wonder whether it is tractable to draw such districtings on real-world graphs. In this section, we argue that the answer is a definitive \emph{yes}. In particular, we find that very simple hill-climbing procedure can yield reasonable districting plans that are highly competitive.

We propose a heuristic hill-climbing procedure based on making local moves called ``single node flips'', and run several experiments on the U.S. states of North Carolina and Arizona: one set seeking to maximize $\delta$-Vote-Band-Competitive districts for $\delta = 0.1, 0.05$ (corresponding to the thresholds used by the Center for Voting and Democracy for ``landslides'' (margin difference above 20\%) and ``competitive'' (margin under 10\%) \cite{2003_DubiousDemocracy}), and one seeking to maximize swing districts. In all cases, we fixed the allowable population deviation at $\eps=5\%$ and required all districts to be contiguous. 

\subsection{Data}
\begin{figure}
    \centering
    \includegraphics*[width=0.19\textwidth]{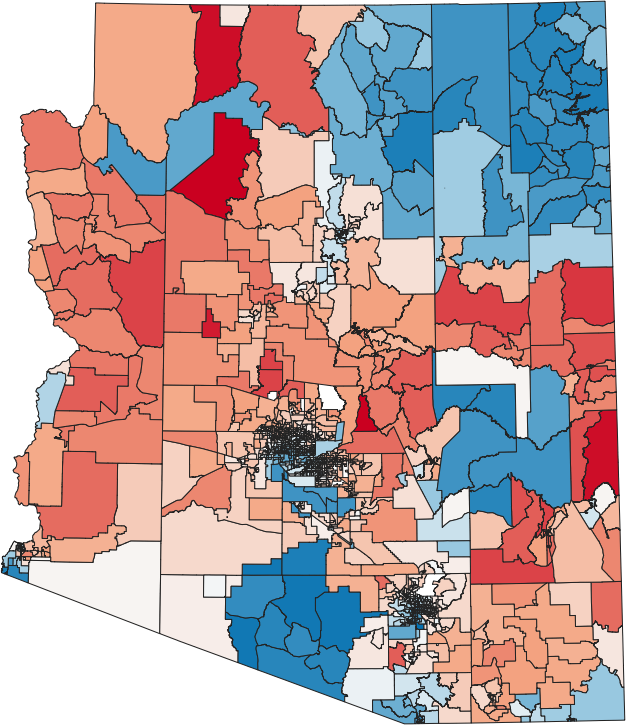}
    \hspace{3em}
    \includegraphics*[width=0.65\textwidth]{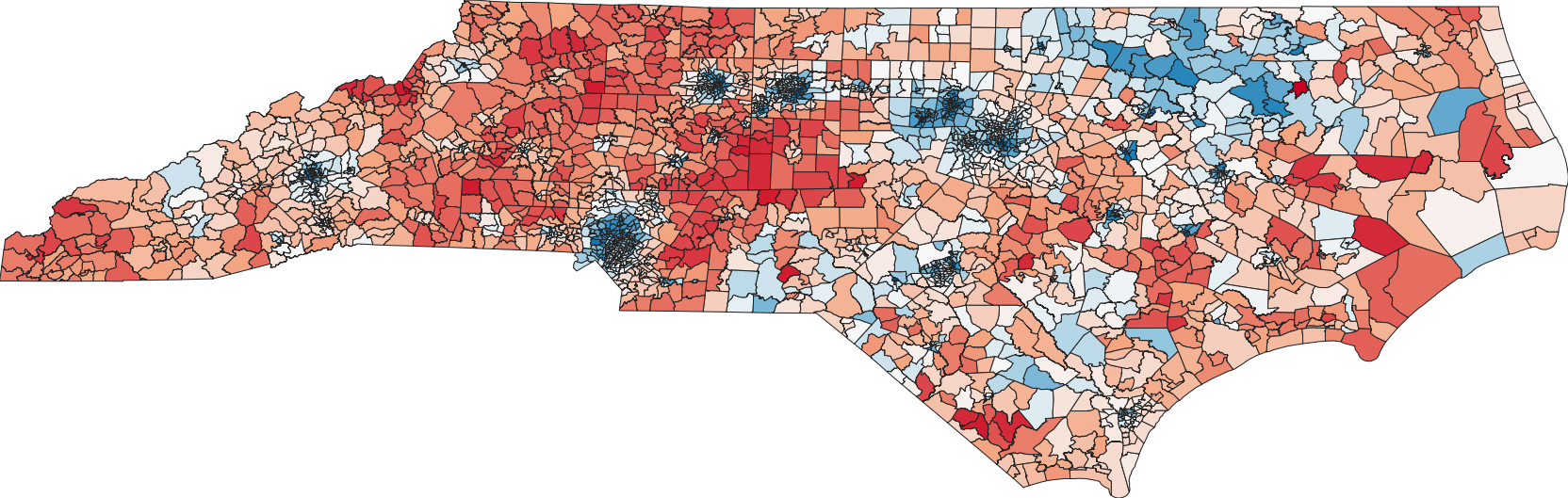}
    \caption{Precinct-level voting data for Arizona and North Carolina. Nodes are colored according to the margin in the 2020 presidential election.}
    \label{fig:AZ_NC_pcts}
\end{figure}

We used precinct-level shapefiles and election data of the U.S. states of Arizona and North Carolina, two medium-sized states that have been highly competitive in recent elections. The geographic data was collected and processed by the Metric Geometry and Gerrymandering Group \cite{_MGGGStates}, and the election data was obtained from the Redistricting Data Hub \cite{_RedistrictingDataHub}. Specifically, the data includes the area, perimeter, and population of each precinct; voting history for the last several elections, and adjacency information. North Carolina has 2,650 precincts and fourteen Congressional districts. Arizona has 1489 precincts and nine Congressional districts. We chose these two states due to their extremely close statewide margins in recent elections, and their moderate size (for reasonable computational burden). They are shown in Figure~\ref{fig:AZ_NC_pcts}.

\paragraph{Voter types}
Evaluating the vote-band competitiveness of a given districting plan depends on the votes cast in a past election. For these experiments, we used the results of the 2020 Presidential election, shown in Fig.~\ref{fig:AZ_NC_pcts}. 
Evaluating whether a given district is swing or not depends on having estimates for the number of reliable voters of each party, and the number of swing voters. The task of estimating voter transitions is a well-studied problem in \emph{Ecological Inference} (EI). We discuss EI and our estimation, using data from the 2012 and 2016 Presidential elections, in Appendix~\ref{apdx:EI}. 

\subsection{Heuristic-based optimization procedure}

In order to find districtings that maximize the number of competitive districts, we use a simple randomized greedy hill-climbing procedure based on repeatedly making ``single node flips'' (also called spin flips) that incrementally improve the plan when measured against some objective(s). These methods are extensively used in ensemble-based analyses of redistricting; see, for example, \cite{CFP_AssessingSignificanceMarkov, CFM_SeparatingEffectSignificance}.\footnote{However, we are \emph{not} doing an ensemble analysis; we will not attempt to sample from a measure or use the Metropolis-Hastings algorithm, as these works do. Rather, we are simply investigating the degree to which a direct optimization can achieve competitive districtings.} 

A \emph{single-node-flip} is a step that flips a node that is on the boundary of two districts from one district to another, subject to some constraints. Specifically, a node $u \in D_j$ is eligible to be flipped to district $D_i$ (where $i\neq j$) if (a) $u$ is adjacent to some $v \in D_i$, (b) $D_j \setminus \{u\}$ remains connected and within population bounds; and (c) $D_i \cup \{u\}$ is within population bounds. 

In each step, we randomly choose a single-node-flip $(u, D_1, D_2)$ from the set of all valid flips, proportional to a score function $J_i$, which quantifies the ``desirability'' of the resulting districts $(D_1 \setminus \{u\}, D_2 \cup \{u\})$, compared to the original districts $(D_1,D_2)$, for particular objectives of interest: 
\[J(u, D_1, D_2) = \exp\left(\sum_{\text{objectives } i} - w_i \bigg(J_i(D_1 \setminus \{u\}) + J_i(D_2\cup \{u\}) - (J_i(D_1) + J_i(D_2))\bigg)\right) \ . \]

Here, low values for $J_i$ are preferred; this weight function weights a flip more heavily the more it decreases the $J_i$ values.\footnote{This type of score function can be considered to be a ``tempered choice'' according to the measure $\pi(D) = \exp\left(\sum_i -w_i J(D_i)\right)$. This form of measure is commonly used in ensemble methods, e.g., in \cite{ACHHM2021_MetropolizedMultiscaleForest, ACHHM2020_MultiScaleMergeSplitMarkov}.}
In this work, our weight includes two score functions: compactness and competitiveness (via swing districts or $\delta$-VBC districts).

We employ the widely-used measure of compactness, the \emph{isoperimetric ratio} \cite{BGHKLMR2017_RedistrictingDrawingLine}, defined as the perimeter squared divided by area. The more compact a district, the lower its isoperimetric ratio, and we define
$J_{iso}(D) = \frac{\text{perimeter}(D)^2}{\text{area}(D)}.$

The vote-band competitiveness term in the score function prioritizes districts that are close to even (i.e., 50-50 margin), and further prioritizes districts where the margin is in the range $\frac12 \pm \delta$: 
\[J_{VBC}(D) = \begin{cases} \left(\frac{PartyA(D)}{Pop(D)} - \frac12\right)^2 & \text{if }\frac{PartyA(D)}{Pop(D)} \not\in \frac12 \pm \delta \ ,\\ 
\frac{1}{16} \left(\frac{PartyA(D)}{Pop(D)} - \frac12\right)^2 & \text{if }\frac{PartyA(D)}{Pop(D)} \in \frac12 \pm \delta \  . \\ 
\end{cases} \]

The swing term in the score function prioritizes having districts where the number of Party A voters plus half the swing voters are close to half of the total population, and further prioritizes districts where neither party's reliable voters comprise more than half of the overall population: 
\[J_{sw}(D) = \begin{cases} \left(\frac12 - \left(\frac12\frac{Swing(D)}{Pop(D)} + \frac{PartyA(D)}{Pop(D)}\right)\right)^2 & \text{if }\frac{PartyA(D)}{Pop(D)} > \frac12 \text{ or } \frac{PartyB(D)}{Pop(D)} > \frac12 \ ,\\ 
    \left(\frac12 - \left(\frac12\frac{Swing(D)}{Pop(D)} + \frac{PartyA(D)}{Pop(D)}\right)\right)^2 \cdot 0.8^2 & \text{if }\frac{PartyA(D)}{Pop(D)} < \frac12 \text{ and } \frac{PartyB(D)}{Pop(D)} < \frac12 \ .\\ 
    
\end{cases}\]

\subsection{Run Parameters}
We ran the hill-climbing procedure for 36,000 steps, restarting from a random initial state every 3,000 steps and taking the most compact plan with the maximal number of competitive districts. We used weights of $w_{iso} = 3, w_{VBC} = 10^5$, for VBC runs and $w_{iso} = 3, w_{sw} = 10^5$ for the Swing runs. We implemented the procedure in Python using the \textsf{gerrychain} package \cite{_GerryChainGerryChainDocumentation}. With unoptimized, single-threaded Python code, the procedure takes about three hours to run on Intel Xeon Gold 6226 2.9 Ghz machines. 

As a basis for comparison, we also ran the hill-climbing procedure for 40,000 steps, restarting every 200 steps, with $w_{iso} = 3, w_{VBC} = 0, w_{sw} = 0$, i.e., prioritizing only compactness. We logged the most-compact plan found in each 200-step interval. We will refer to this set of plans as the \emph{compact sample} below.

\section{Experimental Results} 
\begin{figure}
    \centering 
    \begin{subfigure}{0.31\textwidth}
        \centering 
        \includegraphics*[width=\textwidth]{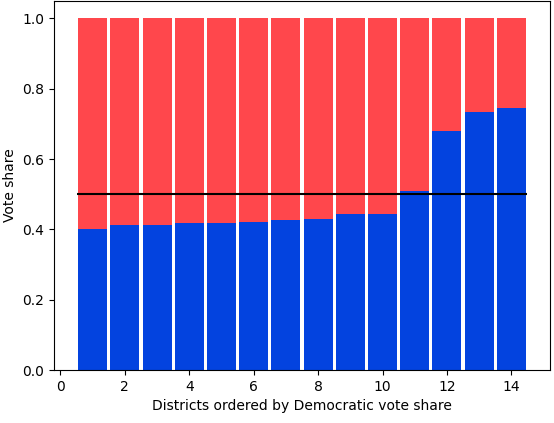}
        \caption{Vote shares of the enacted plan}
    \end{subfigure}
    \hspace{0.5em}
    \begin{subfigure}{0.31\textwidth}
        \centering 
        \includegraphics*[width=\textwidth]{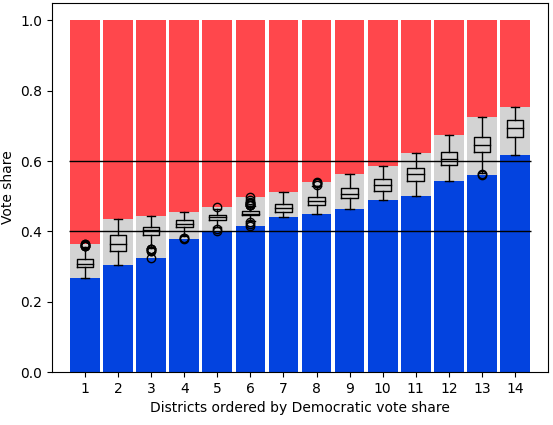}
        \caption{Vote share range of a sample of compact plans}
    \end{subfigure}
    \hspace{0.5em}
    \begin{subfigure}{0.31\textwidth}
        \centering 
        \includegraphics*[width=\textwidth]{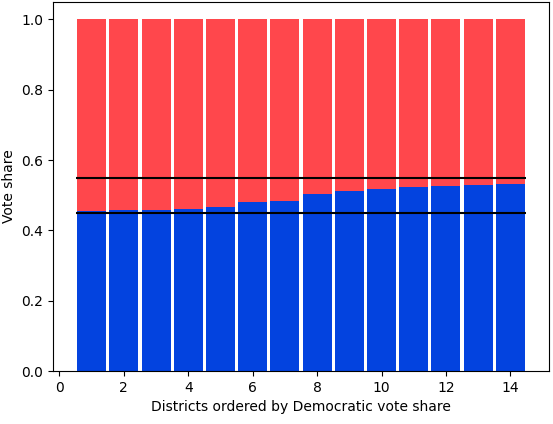}
        \caption{Vote shares of the 5\%-vote-band competitive districting}
    \end{subfigure}
    \caption{Explicitly considering competitiveness allows us to draw plans with significantly fewer safe districts than both the enacted plan and a sample of compact plans.}
    \label{fig:comp_compare}
\end{figure}

We find that the heuristics are extremely effective for constructing districting plans with a significant proportion of swing districts. In fact, we obtain districtings on North Carolina and Arizona where \emph{every single} district is competitive (for 10\%-vote-band, 5\%-vote-band, and swing metrics). 

In Fig. ~\ref{fig:NCresults}, we display the most compact plan for each of $\delta=0.1, 0.05$, and Swing, alongside the current enacted plan, for North Carolina. We also display the vote share of each district: that is, the percentage of Democratic and Republican (and Swing, for the last plot) voters in each district, based on votes cast in the 2020 presidential election. In Fig.~\ref{fig:comp_compare}, we compare the vote share distribution of the 5\%-VBC plan to that of the enacted plan and of the compact sample. 

Notably, the VBC plans exhibit signficantly improved responsiveness: the median district is extremely competitive (49.9\% D-voting) and a small statewide swing in vote share would correspond to a larger change in number of seats won.

\section{Discussion} 

We explicitly do \emph{not} present these plans as examples of \emph{ideal} districting plans. Rather, we present them in contrast to the hardness results presented above, and to explore the consequences of fully prioritizing competitiveness as an objective. Although both swing district and vote-band-competitive district maximization is NP-hard (even if only seeking to make two out of $d$ districts competitive), we explicitly show that it is tractable in practice, even achieving plans that make all $d$ districts competitive. Thus, the point is to dissuade the reader from drawing the conclusion that ``because drawing swing districts is NP-hard, attempting to do so is a lost cause'', or that ``policymakers are exempt from drawing competitive districts because of computational intractibility.'' Instead, drawing competitive districts is very tractable in practice. 

Notably, whereas the enacted plan ``packs'' the heavily-Democratic areas of Charlotte and Raleigh-Durham into as few districts as possible, maximizing competitive districts entails ``cracking'' those areas up into multiple districts. We observe this among all plans that achieved 14 competitive districts; in some sense, this is likely unavoidable given the political geography of the state (in which Democratic voters are heavily concentrated in small geographic areas). 

Indeed, the \emph{ideal} number of competitive districts is almost certainly not the \emph{maximal} number. We leave the question of exactly \emph{how many} districts should be competitive as a normative question (for discussion of the relationship between proportionality and competitive districts, see \cite{DS2022_RedistrictingProportionality}); we simply present the result that it is tractable to achieve anywhere from 0 to $d$ competitive districts on some real-world graphs. 

In particular, while having \emph{every} district in a state be competitive makes the makeup of the state's Congressional delegation highly responsive (that is, small changes in vote share can result in large changes in topline number-of-seats-won), taking this to an extreme (for example, by enacting the plan in Fig.~\ref{fig:NCresults}(c)) can result in a sharp decline in the proportionality of the results. For example, under the plan in Fig.~\ref{fig:NCresults}(c), if one party won 55\% of the votes statewide, they would sweep \emph{all} of the congressional districts - certainly sending an unrepresentative delegation to represent the state.

We show similar results for Arizona in Appendix~\ref{apdx:AZ}. We do not expect all states to admit fully-competitive districtings; indeed, states with remarkably homogeneously distributed electorates (such as Massachusetts) have been observed to be impossible to draw competitive districts on \cite{DGHKNW_LocatingRepresentationalBaseline}. On the other extreme, some states may have voters that are so geographically polarized that drawing competitive districts may require splitting uban areas into unacceptably many districts. For example, Democratic voters in Pennsylvania and Illinois are so heavily concentrated in cities like Philadelphia and Chicago that maximizing competitive districts likely involves splitting them into over a dozen districts - likely an unacceptable result. We leave detailed investigation of these cases for future work.

\begin{figure}
    \begin{subfigure}{\textwidth}
        \centering 
        \includegraphics*[width=0.65\textwidth]{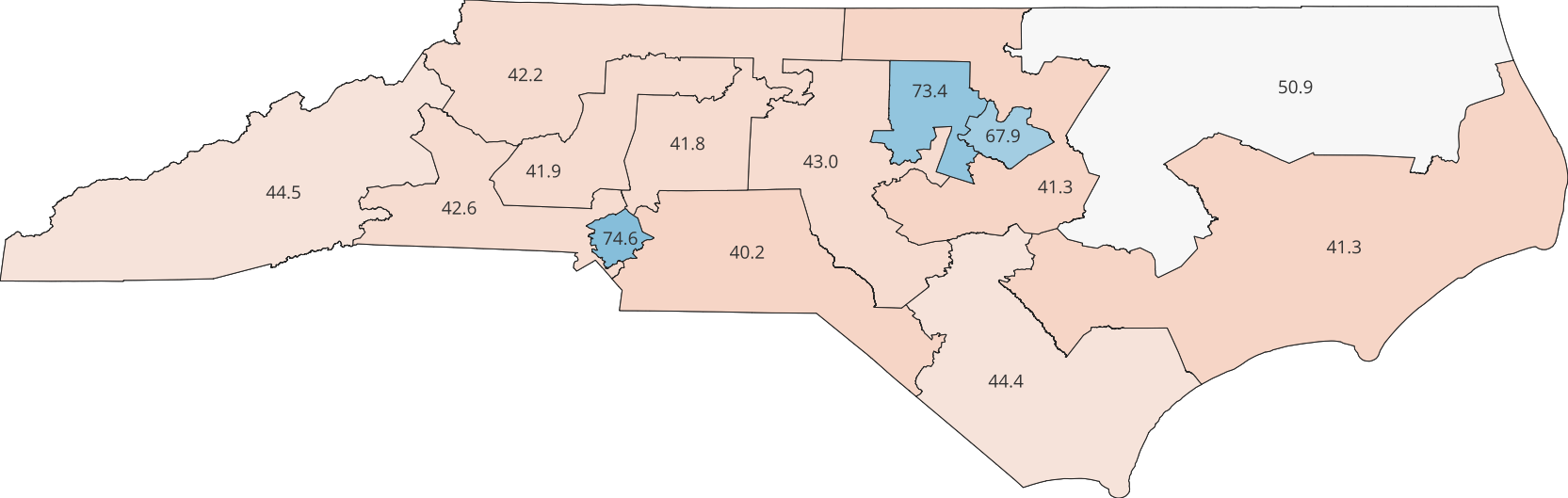}
        \includegraphics*[width=0.3\textwidth]{images/bars/NC_enacted.png}
        \caption{The currently-enacted plan. Despite winning just over half of the two-party vote, Trump would have carried ten of the fourteen Congressional districts.}
    \end{subfigure}

    \begin{subfigure}{\textwidth}
        \centering 
        \includegraphics*[width=0.65\textwidth]{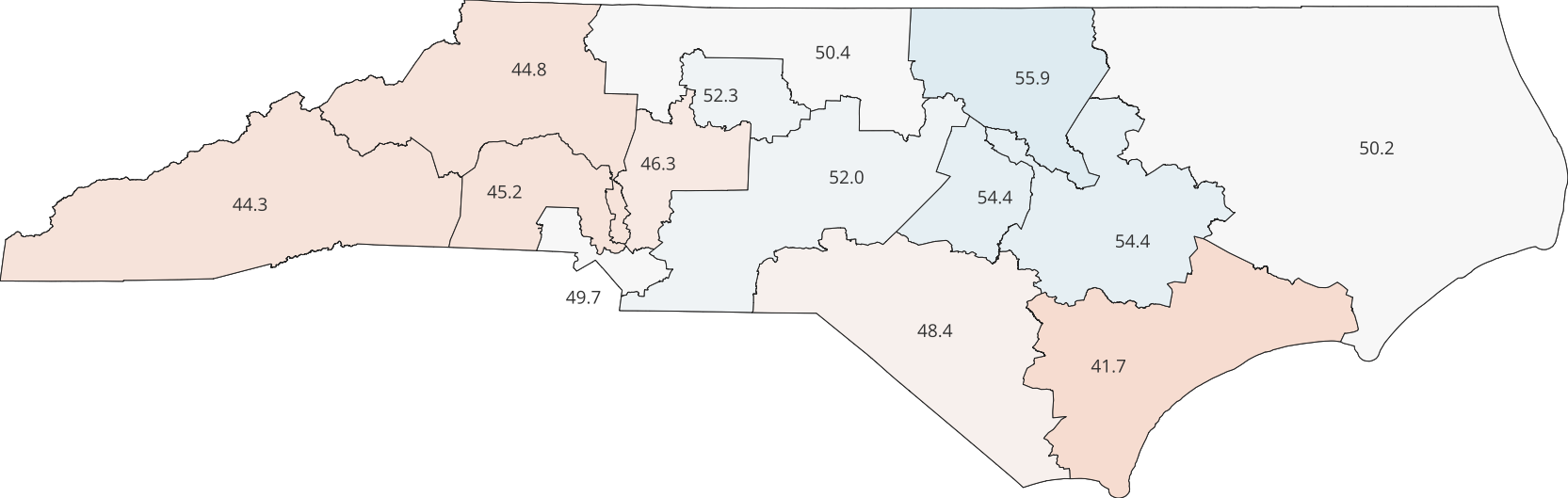}
        \includegraphics*[width=0.3\textwidth]{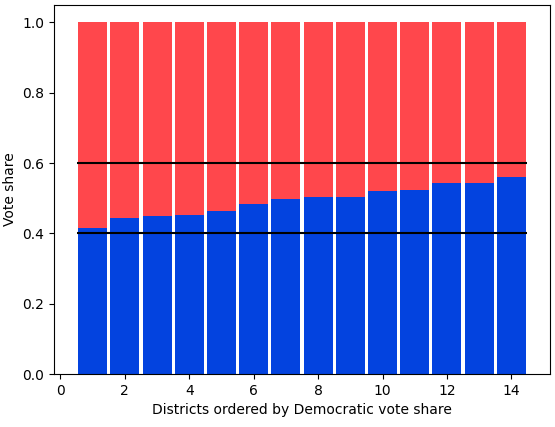}
        \caption{In this plan, all districts are $\delta$-Vote-Band competitive for $\delta=10\%$. That is, all districts have a Biden vote share between 40\% and 60\%. Note that Biden and Trump would have each carried seven of the fourteen districts: a result that is significantly more reflective of the fact that they won nearly the same number of votes.}
    \end{subfigure}

    \begin{subfigure}{\textwidth}
        \centering
        \includegraphics*[width=0.65\textwidth]{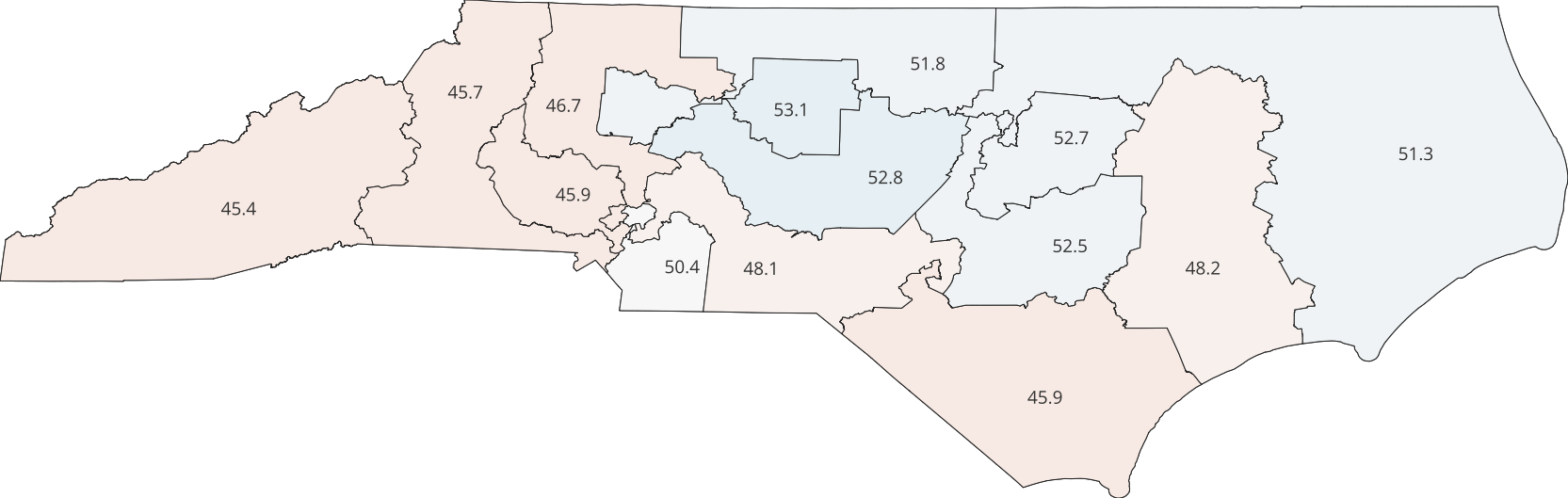}
        \includegraphics*[width=0.3\textwidth]{images/bars/NC_5.png}
        \caption{In this plan, all districts are $\delta$-Vote-Band competitive for $\delta=5\%$: all districts have a Biden vote share between 45\% and 55\%. Again, Biden and Trump would have each carried seven districts. However, the districts are noticeably less compact, with significantly contorted boundary shapes; the heavily Democratic areas of Charlotte and Raleigh-Durham are visibly ``cracked'' among many districts.}    
    \end{subfigure}

    \begin{subfigure}{\textwidth}
        \centering
        \includegraphics*[width=0.65\textwidth]{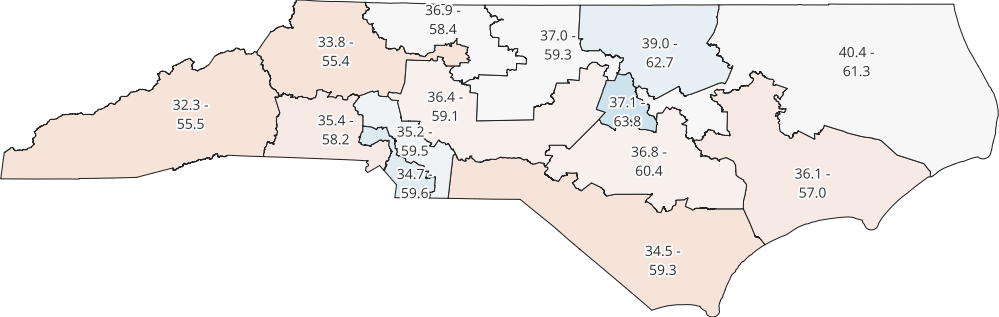}
        \includegraphics*[width=0.3\textwidth]{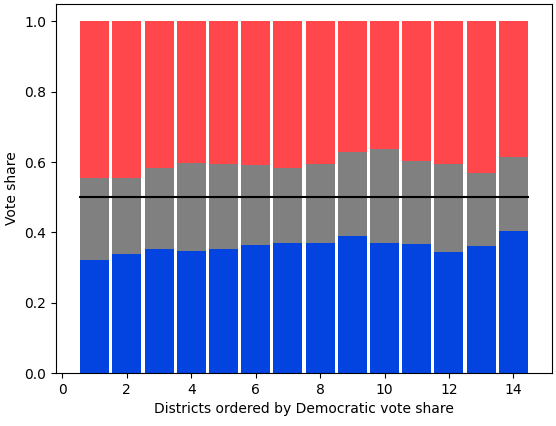}
        \caption{In this plan, all districts are swing. The range of outcomes (ranging from all swing voters voting for Trump to all swing voters voting for Biden) is shown for each district.}    
    \end{subfigure}
    \caption{Our simple hill-climbing procedure successfully finds plans for North Carolina where all fourteen districts are competitive.}
    \label{fig:NCresults}
\end{figure}

\section{Conclusion}

We observe a very large gap between the theoretical intractibility of drawing competitive districts (even on fairly natural instances) and the high performance of empirical heuristics on real instances. This is consistent with the literature on optimizing various metrics in redistricting, as well as with the fact that population-balanced districting itself is clearly achievable in reality while being complexity-wise infeasible in the worst case. We attribute the tractability on real instances to the fact that ``close to optimal'' is an acceptable substitute for ``truly optimal'' in the context of elections, where a large amount of variation and uncertainty is to be expected. 

\newpage

\section{Acknowledgements}

This analysis was conducted using data from the Redistricting Data Hub.

\bibliographystyle{abbrv}
\bibliography{bibliography2}

\newpage

\appendix

\section{Further Proofs}
\label{apdx:proofs}
\subsection{Hardness of Fixed-Proportion Subset}

Here, we present the proof for Theorem~\ref{thm:FPSP-hard}. We reduce from Subset Sum, which is NP-hard even when the inputs can be positive or negative and the desired sum is 0 \cite{KT2006_AlgorithmDesign}.
\begin{proof}
Fix any $\delta \in [0, \frac12)$. Given an instance of Subset Sum $T = [t_1, \cdots t_{n-1}]$ with desired sum 0, let $T' = [t_1, \cdots t_{n-1}, t_n]$, where $t_n = - \sum_{i=1}^{n-1} t_i$, so that $\sum T' = 0$. 

We construct an instance of FPSP($\delta$) as follows. Let $c = -\frac{\min t_i}{\frac12 + \delta}$. Our instance of FPSP will be 
\[S = [(a_1, b_1) \cdots (a_{n+1}, b_{n+1})] \text{ where } a_i = t_i + \left(\frac12 + \delta\right)c \text{ and } b_i = c-a_i \ . \]

Each ``bin'' will have $c$ units total, satisfying the second condition of FPSP. We can verify the first condition holds, using the fact that $\sum t_i=0$:   
\[ \sum_{i=1}^n a_i = \sum_{i=1}^n \left(t_i + \left(\frac12 + \delta\right)c\right) = \sum_{i=1}^n t_i + \left(\frac12 + \delta\right) \sum_{i=1}^n c = 0 + \left(\frac12 + \delta\right) \sum_{i=1}^n (a_i + b_i)\ . \] 

Now, suppose that $S_1, S_2 \subsetneq [n]$ is a solution to FPSP($\delta$) for this instance $S$. Without loss of generality, suppose $n \not\in S_1$. We claim that $\{t_i : i \in S_1 \}$ is a solution to the original Subset Sum instance $T$: 
\[ \sum_{i \in S_1} a_i = \left(\frac12 + \delta\right) \sum_{i\in S_1}(a_i + b_i)\ \implies 
    \sum_{i \in S_1} \left(t_i + \left(\frac12 + \delta\right)c\right) = \left(\frac12 + \delta\right) \sum_{i\in S_1} c \implies 
    \sum_{i \in S_1} t_i = 0   \] 
\end{proof}

\subsection{Hardness of Swing Competitiveness}
\label{apdx:proof_sw}

In this section, we prove Theorem~\ref{thm:SwingMax-hard}, which we restate here for convenience: 

\SwingMaxHard* 

The proof is nearly identical the the proof of Theorem~\ref{thm:VBC-hard}. 

\begin{proof}
    Fix $\eps$ and let $\delta=0$. Let $S = [(a_1, b_1) \cdots (a_n, b_n)]$ be an instance of FPSP(0), and let $Z = \sum_{i=1}^n (a_i + b_i)$. Our redistricting instance (``state'') for Swing-Max($\eps$) will be a $3\times (n+2)$ grid of cells $C_{i,j}$, where: 
    \begin{itemize}
        \item All cells in the first and third row have zero population: $pop(C_{1,\cdot}) = pop(C_{3, \cdot}) = 0$.
        \item For second-row cells $C_{2,i}$ for $i=1, 2, \cdots n$, 
        \[pop(C_{2,i}) = (a_i + b_i) \qquad Party_A(C_{2,i}) = a_i  \qquad Swing(C_{2, i}) = b_i \]
        \item Finally, for second-row cells $C_{2,j}$ for $j = n+1, n+2$, 
        \[pop(C_{2,j}) = \frac12 P, \text{ where } P= \frac{\frac12 + \eps}{\frac12 - \eps} \left(1 + Z\right)\] 
        \[Party_A(C_{2,j}) = \frac14 P \qquad  Swing(C_{2,j}) = \frac14 P\]
    \end{itemize}

    The total population of the state is 
    \[ Pop_{total} = 0 + Z + 2 \left(\frac12 P\right) = Z+P = \frac{\frac12 - \eps + \frac12 + \eps}{\left(\frac12 - \eps\right)}Z +\frac{\frac12 + \eps}{\left(\frac12 - \eps\right)} = \frac{1}{\frac12 - \eps}Z + \frac{\frac12 + \eps}{\frac12 - \eps} \]
    
    Let $\{D_1, D_2\}$ be a $\epsilon$-valid Swing districting on this state. 
    
    The following observations still hold, from the proof of Theorem~\ref{thm:VBC-hard}:

    \textbf{Observation 0}: One can district this instance in polynomial time.
    
    \textbf{Observation 1}: Cells $C_{2,n+1}$ and $C_{2,n+2}$ must be assigned to different districts. Without loss of generality, let $C_{2, n+1} \in D_1$.
        
    \textbf{Observation 2}: Each district must contain at least one of $C_{2,i}$ for $i \in [n]$. 
    
    Claim: $S_1 = \{i : C_{2, i} \in D_1\}, S_2 = \{i : C_{2, i}\in D_2\}$ is a valid solution to the FPSP($\eps$) instance. 
    
    For $S_1, S_2$ to be a solution to FPSP($\eps$), we must show (a) both are nonempty (which follows from Observation 2 above), and (b) $\sum_{i \in S_1} a_i = \left(\frac12 + \delta\right) \sum_{i \in S_1} (a_i + b_i)$.
    
    $\frac12$ of the overall population is Party A voters: 
    \begin{align*} \frac{\text{total }Party_A}{\text{total population}} &= \frac{\sum_{i=1}^n a_i + \left(\frac12\right)\frac{\frac12+\eps}{\frac12-\eps}(1+Z)}{\sum_{i=1}^n (a_i + b_i) + \frac{\frac12+\eps}{\frac12-\eps}(1+Z)} \\ 
    &= \frac{(\frac12)\left(\sum_{i=1}^n (a_i+b_i) +\frac{\frac12+\eps}{\frac12-\eps}(1+Z)\right)}{\sum_{i=1}^n (a_i + b_i) + \frac{\frac12+\eps}{\frac12-\eps}(1+Z)} = \frac12 \end{align*}
    
    Since both districts' margins are exactly $\frac12$, each one must have \emph{exactly} $\frac12$ Party A voters. So, 
    \begin{align*}
        \frac12 &= \frac{Party_A(D_1)}{Party_A(D_1) + Party_B(D_1)} \\ 
        &= \frac{\left(\frac12\right) \frac12 P + \sum_{i \in S_1} a_i}{\frac12 P + \sum_{i \in S_1} (a_i + b_i)}\\ 
        \left(\frac12\right) \left(\frac12 P + \sum_{i \in S_1} (a_i + b_i)\right) &= \left(\frac12 \right)\frac12 P + \sum_{i \in S_1} a_i\\ 
        \left(\frac12\right) \frac12 P + \left(\frac12\right)\sum_{i \in S_1} (a_i + b_i) &= \left(\frac12\right)\frac12 P + \sum_{i \in S_1} a_i\\ 
        \sum_{i \in S_1} a_i &= \left(\frac12\right) {\sum_{i \in S_1} (a_i + b_i)}
    \end{align*}
    
    as required. Thus, $S_1, S_2$ are a valid solution to the FPSP instance. 
\end{proof}

\subsection{Hardness for an arbitrary number of districts}
\label{apdx:proof_dgeq2}

In this section, we prove Corollary~\ref{cor:dgeq2}: 

\MoreDistricts* 

\begin{proof}
    We reduce from an instance with $d=2$. Fix $d', k$ with $2 \leq k \leq d'$. We simply add $d'-2$ cells to the instance. Each cell will have the maximal allowable district population. $k-2$ of the districts will have half Party A voters and half Party B voters (which makes that singleton district competitive under either Swing or $\delta$-VBC definitions), and the remainder will be entirely Party A voters (i.e., uncompetitive). 

    Any valid districting with $k$ competitive districts must have the added cells as singleton districts, of which $k-2$ will be competitive. The other two districts must be a competitive districting of the original instance. 
\end{proof}

\section{Ecological Inference for Estimating Swing Voters} 
\label{apdx:EI}

The task of estimating voter transitions from election data is a well-studied problem in Ecological Inference (EI). Given only top-level voting information (i.e., the number of votes cast for each candidate in two subsequent elections), the task of finding the number of voters who switched their vote from one election to the next (or who voted in one election and not the other) is highly underdetermined. As a result, EI techniques have no worst-case guarantees; however, they have been shown to perform well in practice. For an overview, see \cite{KTMSK2016_EstimationVoterTransitions}. 

The task is made easier by the fact that we have the marginals (i.e., total votes cast) for each precinct in the state. \textsf{nslphom} is a multi-iteration Linear Programming technique that takes advantage of this fact, using statewide homogeneity assumptions, developed by Pavia et al. \cite{PR2022_ImprovingEstimatesAccuracy}. It is available as an R package, which we used to estimate the inner cell values of the $3\times 3$ tables of the form shown in Fig.~\ref{fig:EI} (one table per precinct). For example, $71.30$ is the estimate for the number of voters who voted for the Democratic candidate in 2012 but did not vote in 2016.

\begin{figure}[h]
    \centering
    \begin{tabular}{|cr|ccc|c|}
        \hline 
            & & \multicolumn{3}{c|}{2016 votes} & \\ 
            & & Democratic & Republican & Nonvote/Other & Total \\ 
        \hline 
        \multirow{3}{*}{\makecell{2012\\votes}} 
            & Democratic    & 407.45 & 9.24 & 71.30 & 488\\ 
            & Republican    & 3.55 & 1583.69 & 73.76 & 1660\\ 
            & Nonvote/Other & 0.00 & 272.07 & 2713.93 & 2986\\
        \hline 
        & Total & 410 & 1864 & 2859 & \\
        \hline
    \end{tabular}
    \caption{The EI table for votes cast in Alamance County, North Carolina, Precinct 01. The inner cell values are estimated using the R package \textsf{nslphom}, using the known marginal values. In this case, the estimated number of swing voters is $3.55 + 9.24 + \frac12(71.30 + 73.76 + 0.0 + 272.07) = 221.36$}
    \label{fig:EI}
\end{figure}

We use the election results of the 2012 and 2016 elections to estimate swing voter counts and the results of the 2020 election to evaluate the final redistricting maps in order to keep estimation and evaluation metrics separate. For each precinct, we compute the final estimate for ``swing voters'' by summing the off-diagonal elements of the EI table, with the entries in a ``Nonvote/Other'' row or column halved. Intuitively, compared to a voter who switches their vote from Party A to Party B, a voter who merely goes from abstaining to voting (or from voting to abstaining) only changes the top-level margin by half as much. 

\section{Arizona Experimental Results}
\label{apdx:AZ}
\begin{figure}[H]
    \begin{subfigure}{\textwidth}
        \centering 
        \includegraphics*[width=0.2\textwidth]{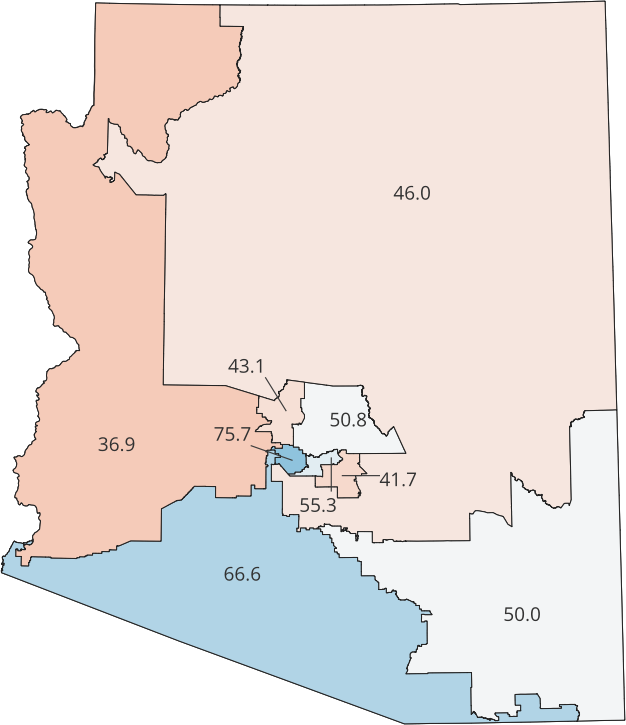}
        \hspace{3em}
        \includegraphics*[width=0.32\textwidth]{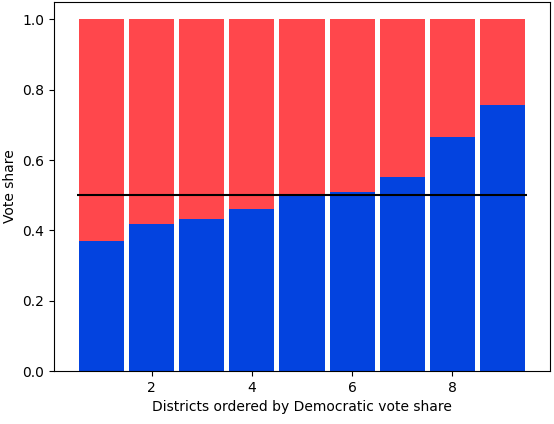}
        \caption{The current enacted plan.}
    \end{subfigure}
    \begin{subfigure}{\textwidth}
        \centering 
        \includegraphics*[width=0.2\textwidth]{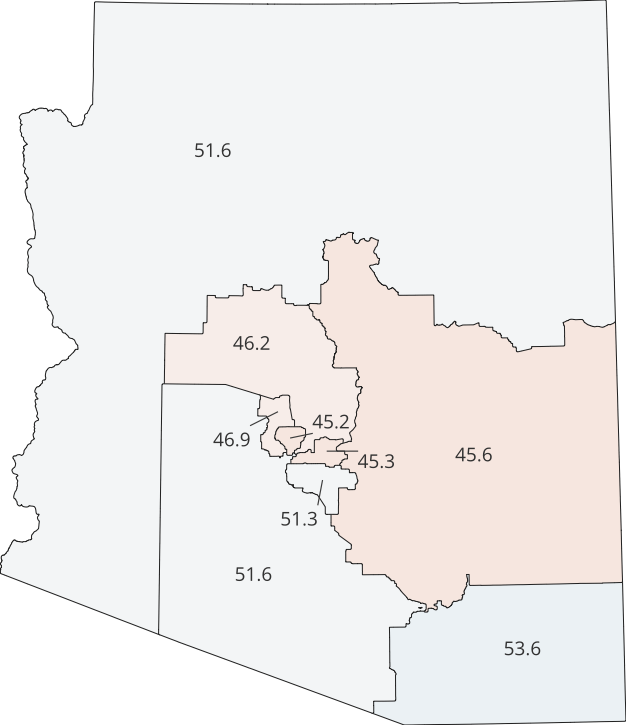}
        \hspace{3em}
        \includegraphics*[width=0.32\textwidth]{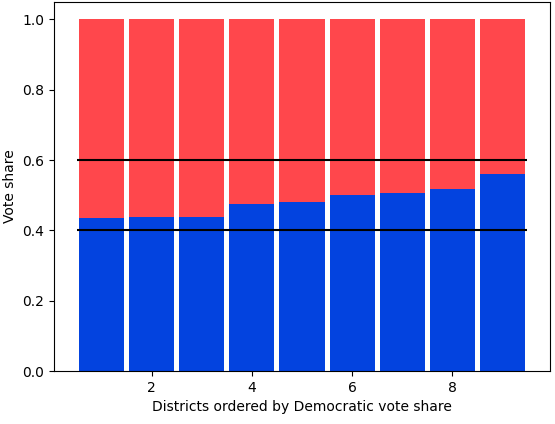}
        \caption{A 0.1-Vote-Band-Competitive districting: all districts have a Biden vote share between 40\% and 60\%.}
    \end{subfigure}
    \begin{subfigure}{\textwidth}
        \centering
        \includegraphics*[width=0.2\textwidth]{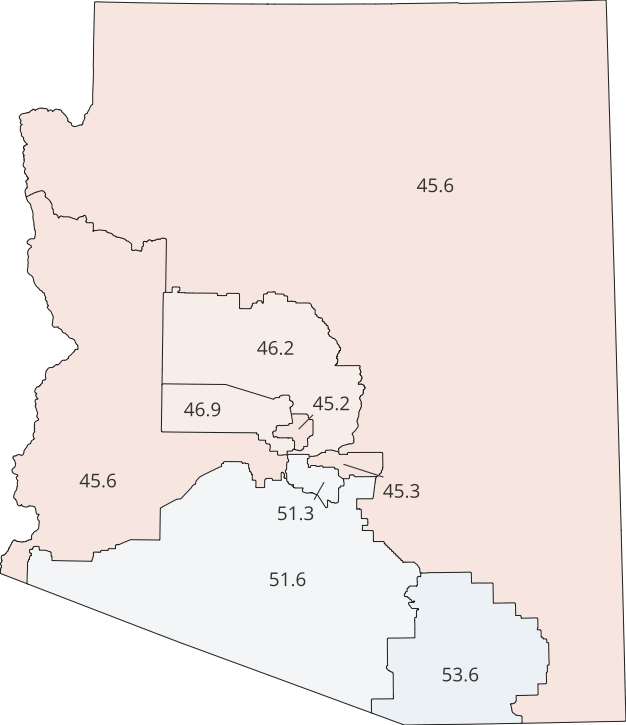}
        \hspace{3em}
        \includegraphics*[width=0.32\textwidth]{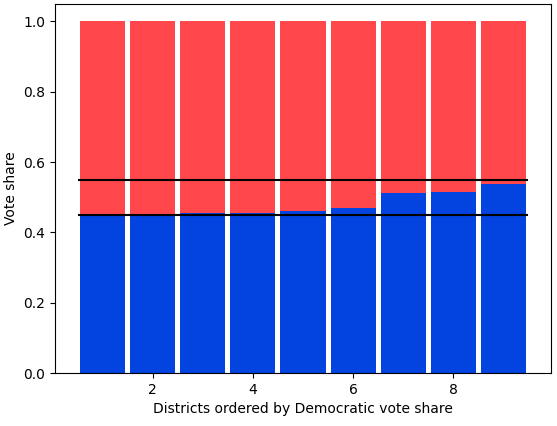}
        \caption{A 0.05-Vote-Band-Competitive districting: all districts have a Biden vote share between 45\% and 55\%. Unlike in North Carolina, this appears to be achievable without significant loss of compactness.}    
    \end{subfigure}
    \begin{subfigure}{\textwidth}
        \centering
        \includegraphics*[width=0.2\textwidth]{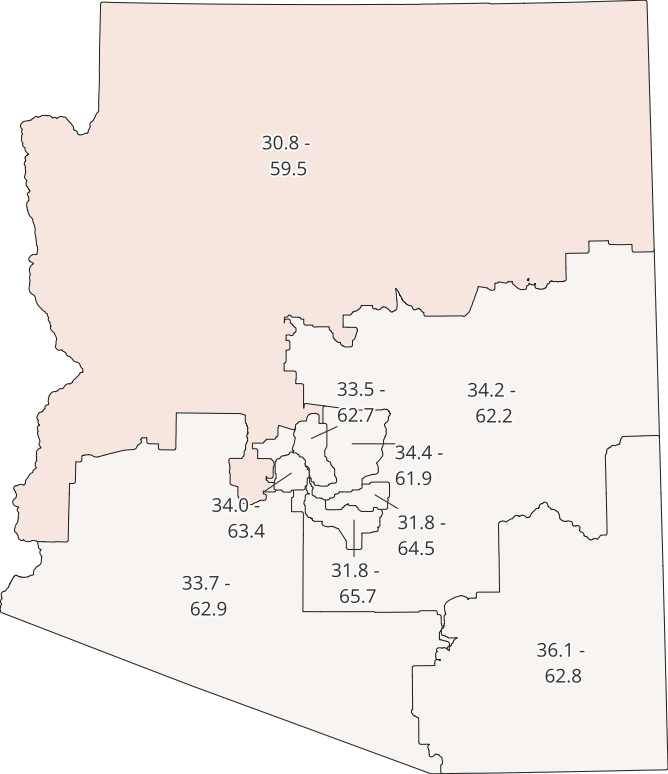}
        \hspace{3em}
        \includegraphics*[width=0.32\textwidth]{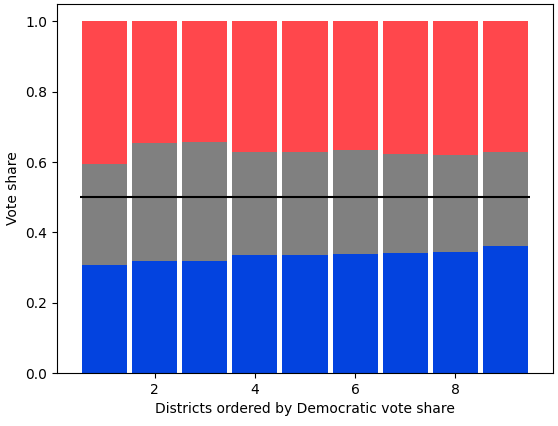}
        \caption{In this plan, all districts are swing. The range of outcomes (ranging from all swing voters voting for Trump to all swing voters voting for Biden) is shown for each district.}    
    \end{subfigure}
    \caption{Our hill-climbing procedure successfully finds plans for Arizona where all nine districts are competitive. Unlike in North Carolina, there is no significant sacrifice in compactness to achieve a $\delta=5\%$-vote-band-competitive districting.}
    \label{fig:AZresults}
\end{figure}

\section{Polynomial Algorithms and Approximations}
\label{apdx:poly}

In this section, we consider special cases of the Swing District Maximization problem in which either the graph $G$ has a special structure (line, bounded degree tree, etc), the population constraint is relaxed, or the districts satisfy an additional structure.

\subsection{Line}
We start by considering the case where $G$ is the line graph on $n$ vertices $c_1, \ldots, c_n$, such that there is an edge between $c_i$ and $c_{i+1}$ for $1 \leq i \leq n-1$. Satisfying the connectivity constraint in this case is easy; simply make sure that every districts has only consecutive cells $c_j, c_{j+1}, \ldots, c_{i}$. A district is feasible if it's connected and satisfies the population constraints. We  can find the optimal $d$-districting by solving a dynamic program. For fixed $d$, $i \in \{1,\ldots,n\}$ and $k \leq d$, let $M(i,k)$ denote the maximum number of swing districts among all valid $d$-districtings of the subgraph induced by $c_1, \ldots, c_i$. We show that 
\begin{equation}\label{eq:dp_line}
M(i,k) = \max\limits_{\substack{1 \leq j < i \\ (c_{j+1},\ldots, c_i) \mbox{ feasible}}} \Big\{ M(j,k-1) + \mathbbm{1}\{ (c_{j+1},\ldots, c_i) \mbox{ is swing}\} \Big\},\end{equation} 
and

\begin{align*}
    M(i,0) &= -\infty & & \forall \ i \in \{1,\ldots,n\}  \\ 
    M(i,1) &=
        \mathbbm{1}\{ (c_1,\ldots, c_i) \mbox{ is swing}\} & \mbox{if  } (c_1,\ldots, c_i) \mbox{ feasible, } & \forall \ i \in \{1,\ldots,n\}\\
    M(i,1) &=
        -\infty & \mbox{if  } (c_1,\ldots, c_i) \mbox{ is not feasible, } & \forall \ i \in \{1,\ldots,n\}
\end{align*}

Equation \eqref{eq:dp_line} holds since to get the maximum number of swing districts in a $k$-districting of $c_1, \ldots, c_i$, we need to decide on the cell $c_{j+1}$ that limits the last district from the left. Once the last  district $\{ (c_{j+1},\ldots, c_i)\}$ is fixed, we need to pick the remaining $k-1$ districts from $c_1,\ldots, c_j$. Checking if a subset of cells satisfies population constraints and induces a swing district can be done in $O(n)$ time. Therefore, we can fill the $n \times d$ entries of the matrix $M(i,k)$ top down from left to right. The optimal solution is stored in $M(n,d)$. We have the following lemma.
\begin{lemma}
If $G$ is a line on $n$ cells, we can compute the optimal $d$-districting in $O(n^2d)$ time.
\end{lemma}

\subsection{Bounded-Degree Trees with Districts of Bounded Depth}
In this subsection, we consider bounded-degree trees with the additional assumptions that districts need to have bounded depth, that is, the distance between every two cells in the same district must be less than a parameter $d > 0$. Consider the graph $G$ to be a tree with $n$ vertices and a maximum degree $\Delta$. Let $\eps$ be the population tolerance. We require that in a valid districting, all districts must run for a depth of at most $d$ (i.e., the diameter of every district is less than $d$). \\

Let $\mathcal{D}(v,\eps,d)$ be the set of $\eps$-valid districts with depth at most $d$ that are rooted at v.

\begin{claim}\label{claim:dp_tree}
\[ |\mathcal{D}(v,\eps,d)| \leq 2^{\Delta^d}\]
\end{claim}
\begin{proof}
The number of vertices at a distance less or equal than $d$ from the root $v$ is less than $\Delta^d$. For every one of these vertices and every $D \in \mathcal{D}(v,\eps,d)$, $D$ can either contain the vertex or not.
\end{proof}

When $d = O(1)$ and $\Delta = O(1)$, we propose a polynomial time dynamic program to solve the swing district maximization problem. For a fixed $d$ and $k \leq d$, let $M(v,k)$ denote the maximum number of swing districts in a $\eps$-valid $k$-districting for the subtree of $G$ that is rooted at $v$. To get the optimal districting, we need to first fix the district that $v$ will belong to in $\mathcal{D}(v,\eps,d)$. Because of the depth constraint on districts, the number of possible choices is bounded by $2^{\Delta^d}$. After we fix the district of $v$, we need to choose the roots of the remaining $k-1$ districts.\\

Let $G(v)$ be the subtree of $G$ that is rooted at $v$. Let $D \in \mathcal{D}(v,\eps,d)$, and let $\mathcal{R}(D)$ be the roots of the subtrees of $G(v) \setminus D$ (See Figure \ref{fig:bounded_tree}). If $|\mathcal{R}(D)| > k-1$, then clearly we cannot assign all the remaining cells to districts. Therefore we need $0 <|\mathcal{R}(D)| \leq k-1$. Furthermore, to compute the remaining $k-1$ districts, we need to start from the roots in $\mathcal{R}(D)$ such that, every vertex in $\mathcal{R}(D)$ will give rise to at least one district. To decide how many district every tree rooted at a vertex in  $\mathcal{R}(D)$ needs to have, we assign a number $\ell(u) \in \{1, \ldots, k-1\}$ for every $u \in \mathcal{R}(D)$, such that $\sum\limits_{u \in \mathcal{R}(D)} \ell(u) = k-1$, and the subtree rooted at $u$ contains $l(u)$ out of the remaining $k-1$ districts. This gives rise to the following dynamic program

\begin{align*}
    M(v,k) & =   \max\limits_{\substack{D \in \mathcal{D}(v,\eps,d)\\ |\mathcal{R}(D)| \leq k-1 \\ \sum\limits_{u \in \mathcal{R}(D)} \ell(u) = k-1 }}\Big\{ \mathbbm{1}\{D \mbox{ is swing}\} + \sum\limits_{ u \in \mathcal{R}(D)} M(u,\ell(u))\Big\}
\end{align*}

If $r$ is the root of $G$, the $M(r,d)$ will contain the optimal number of swing districts. In order to get $M(r,d)$, we need to top-down fill $O(nk)$ entries of $M$. To fill an entry $M(v,k)$, we have to choose a district $D \in \mathcal{D}(v,\eps,d)$ and an assignment $u\mapsto \ell(u)$ for $u \in \mathcal{R}(D)$ such that $\sum\limits_{u \in \mathcal{R}(D)} \ell(u) = k-1$. 

\begin{claim}\label{claim:R(D)}
Once a district $D \in \mathcal{D}(v,\eps,d)$ is fixed, the number of possible assignments $u\mapsto \ell(u)$ for $u \in \mathcal{R}(D)$ such that $\sum\limits_{u \in \mathcal{R}(D)} \ell(u) = k-1$ is less than $k^{ 2^{{\Delta}^d}}$.
\end{claim}
\begin{proof}
Similarly to the proof of Claim \ref{claim:dp_tree}, we can show that $|\mathcal{R}(D)| \leq 2^{\Delta^d}.$ The number of  positive assignments $u\mapsto \ell(u)$ such that $1 \leq \ell(u) \leq k-1$ is less than $k^{2^{{\Delta}^d}}.$
\end{proof}

The combination of Claim \ref{claim:dp_tree} and \ref{claim:R(D)} show that every entry $M(v,k)$ can be computed in $O(2^{{\Delta}^d}k^{2^{{\Delta}^d}})$ time given that we already know the previous entries. We therefore have the following lemma.
\begin{lemma}
If $G$ is a tree with a maximum degree $\Delta$, and districts can have at most a depth of $d$, we can compute the optimal $d$-districting in $O(n2^{{\Delta}^d}k^{1+2^{{\Delta}^d}})$ time.
\end{lemma}

\begin{figure}
    \centering
    \includegraphics[scale=0.55]{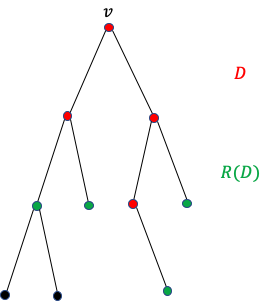}
    \caption{The vertices in red depict the district $D$ that is rooted at $v$. The vertices in green depict $\mathcal{R(D)}$.}
    \label{fig:bounded_tree}
\end{figure}
\subsection{Convexity Districtings of Grid Graphs}
In this subsection, we assume that the graph $G$ is an $m \times n$ grid. If we require all the districts to be $x$-convex, that is, if two cells $c_1$ and $c_2$ of the same row  are assigned to the same district, then all the cells between $c_1$ and $c_2$ of that same row are also assigned to that district. This case encompasses a compactness constraint since convexity has been used in gerrymandering studies as a measure of compactness to examine how redistricting reshapes the geography of congressional districts \cite{CDPAS2018_AlleviatingPartisanGerrymandering}.

\begin{restatable}{thm}
Let $G$ be an  $m \times n$ grid, and let $P = Pop(G)$ be the total
population on $P$. There exists an algorithm for computing an $x$-convex valid $d$-districting of $G$ with maximum swing districts, with running time $(Pm)^{O(d)}$. In particular, the running time is polynomial when the total
population is polynomial and the total number of partitions is a constant.
\end{restatable}

\begin{proof}
Let $D_1, \ldots, D_k$ be an $x$-convex $d$-districting of $G$. For any
$i \in \{1, \ldots, n\}$, let $C_i$ be the i-th column of $G$. We observe that for all $i \in \{1, \ldots, n\}$, and for all $j \in \{1, \ldots, d\}$, we have that $D_j \cap C_i$ is either empty, or consists of a single rectangle of width 1. Let $\mathcal{C}_i$ be the
set of all \textit{contiguous} partitions of $C_i$ into exactly $d$ (possibly empty) segments, each labeled with a unique integer in
$\{1, \ldots, d\}$. We further define
\begin{align*}
    \alpha_{i,j} & = \mbox{PartyA}\big(D_j^* \cap (C_1 \cup \ldots \cup C_i)\big)\\
    \beta_{i,j} & = \mbox{PartyB}\big(D_j^* \cap (C_1 \cup \ldots \cup C_i)\big)\\
    \gamma_{i,j} & = \mbox{Swing}\big(D_j^* \cap (C_1 \cup \ldots \cup C_i)\big),\\
\end{align*}
For each column $i \in [n]$, and each district $j \in [d]$, $\alpha_{i,j}$ (resp. $\beta_{i,j}$, $\gamma_{i,j}$) denote the number of Party A (resp. Party B, swing) voters in district $D_j$ between column 1 and column $i$.\\

We can enumerate all the possible solutions starting from the first column and moving to the right as follows. For each $i \in \{1, \ldots, n\}$, let $I_i = \mathbb{N}^{3d} \times \mathcal{C}_i \times [m]^{d}$. Let $$X_i = (\alpha_{i,1}, \beta_{i,1}, \gamma_{i,1}, \ldots, \alpha_{i,d}, \beta_{i,d}, \gamma_{i,d}, \Z_i, \F_i) \in I_i,$$

where $\Z_i$ is a $d$-partition of the column $C_i$ and $\F_i = \{F_{i,1}, \ldots, F_{i,d}\}$ is a collection of the forbidden indices for every district in the next column $C_{i+1}$, to ensure the connectivity of the districts as well as the $x$-convexity constraints.\\

If $i = 1$,  we say that $X_i$ is feasible if $\F_1 =\{ \emptyset, \ldots \emptyset\}$ and, for all $j \in \{1, \ldots, d\}$, the set $\Z_1 = \{Z_1, \ldots, Z_{d}\}$  satisfies
\[\alpha_{i,j} =  \mbox{PartyA}\big(Z_{i,j}\big), \quad  \beta_{i,j} = \mbox{PartyB}\big(Z_{i,j}\big), \quad \mbox{and} \quad \gamma_{i,j} = \mbox{Swing}\big(Z_{i,j}\big)\]

If $i > 1$, we say that $X_i$ with $\Z_{i} =  \{Z_{i,1}, \ldots, Z_{i,d}\}$ is feasible if the following holds, there exists some $X_{i-1} = (\alpha_{i-1,1}, \beta_{i-1,1}, \gamma_{i-1,1}, \ldots, \alpha_{i-1,d}, \beta_{i-1,d}, \gamma_{i-1,d}, Z_{i-1}, \F_{i-1}) \in I_{i-1}$ such that 

\begin{gather}
\alpha_{i,j} = \alpha_{i-1,j}+ \mbox{PartyA}\big(Z_{i,j}\big), \quad  \beta_{i,j} = \beta_{i-1,j}+\mbox{PartyB}\big(Z_{i,j}\big), \quad \mbox{and} \quad \gamma_{i,j} = \gamma_{i-1,j} + \mbox{Swing}\big(Z_{i,j}\big), \label{cvx:c2}
\\
 \Z_{i,j} \cap \F_{i-1,j}  = \emptyset, \quad \label{cvx:c3} \\
 \mbox{If } Z_{i,j} \not\subset Z_{i-1,j} \mbox{ then } F_{i,j}  = (Z_{i-1,j} \setminus Z_{i,j}) \cup \F_{i-1,j}, \mbox{ else } F_{i,j} = \F_{i-1,j}\label{cvx:c4}
\end{gather}


The constraint \eqref{cvx:c2} simply states that the voter populations in columns $1,\ldots, i$ are equal to the the voter populations in columns $1,\ldots, i-1$ plus the voters from partition $\Z_i$. Constraint \eqref{cvx:c3} ensure that the partition of the $i$-th column has to respect $x$-convexity and not include any of the forbidden cells from $\F_{i-1}$. The last constraint \eqref{cvx:c4} shows how to update the forbidden cells for the next column $i+1$. If for a district $D_j$, the indices of the rows added from column $i$ to $D_j$ is included in the set of rows of $D_j$ from column $i-1$, then there is no need to add any other forbidden row for the next column. If however, $\Z_{i,j} \not\subset\Z_{i-1,j}$, that means that the column $i$ does not ``transfer'' the rows of $D_j$ from column $i-1$ to column $i$, then the ``non-transfered'' rows have to be forbidden in the next column. \\

For each $i \in \{1, \ldots, n\}$ we inductively compute the set of all feasible $X_i \in I_i$. This can be done
in time $(Pm)^{O(d)}$ where $P$ is the total population of the map. 
\end{proof}

\subsection{MWIS $O(1)$-Approximation for $d = \Theta(n)$}
The previous subsection presents an optimal solution in polynomial time if the districts must satisfy $x$-convexity, the number of districts is constant, and the total population is polynomial in the input. If however, $d$ is no longer constant, then the previous theorem no longer holds. In this subsection, we focus on the case where the number of districts is linear in the number of cells , that is $d = \Theta(n)$. The population constraint (C2) relaxes the ideal $d$-equipartition constraint for a district
$\{D_1, \ldots, D_{d}\}$, that is,
\begin{equation*}
    \forall \ j: Pop(D_j) \in \{ \lfloor Pop(G)/d \rfloor, \lceil Pop(G)/d \rceil \}.
\end{equation*}\\
In this subsection, we relax (C2) to count the number of cells per district instead of population per district, i.e.
\begin{equation*}
    \mbox{(C2-bis)} \ \  \forall \ j: |D_j| = n/d.
\end{equation*}
The constraint (C2-bis) focuses on the number of cells per district as a proxy for the population. This is a realistic relaxation if the population is uniformly distributed across the graph, or if we would like districts to be similar in term of some attribute that is uniform over the cells (area, budget, etc.).\\

If $d = \Theta(n)$ and that we require every district to have exactly $n/d$ cells, we can compute all the possible districts (at most $\binom{n}{n/d} = n^{O(1)}$) and construct a conflict graph $\mathcal{C}$, where nodes represent feasible districts, and two nodes are connected if their corresponding districts share at least one cell. Since each district is required to have exactly $n/d$ cells, then we can show that the conflict graph $\mathcal{C}$ has a special structure, in the sens that no node in $\mathcal{C}$ has more than $n/d + 1$ distinct independent neighbours, i.e., neighbours that do not have any edge in between them.

\begin{lemma}
Let $d = n/d$. The conflict graph $\mathcal{C}$ is $d+1$-claw free, i.e., no node in $\mathcal{C}$ has $d+1$ distinct independent neighbors.
\end{lemma}
\begin{proof}
Assume that there is a node $D$ in the conflict graph that has $d+1$ independent neighbours $N_1, \ldots, N_{d+1}$. This implies that the district $D$ shares at least one cell with every district $N_i$ for $i \in \{1,\ldots, \Delta\}$, but $N_i$ and $N_j$ do not share any cell if $i \neq j$. Therefore, there are at least $d + 1$ cells in $D$, one for every intersection $D \cap N_i$. This is a contradiction with the fact that number of cells in every district is equal to $d$.

\end{proof}

Finally, we show how we can formulate the Swing District Maximization problem as a Maximum Weighted Independent Set problem on the conflict graph $\mathcal{C}$. We set the weights of the nodes of $\mathcal{C}$ as follows: $1$ if the district is swing, 0 otherwise. Therefore, a $d$-districting in $G$ with maximum number of swing districts corresponds to a maximum weight independent set in $\mathcal{C}$ with maximum weight. Notice that since the districts in $\mathcal{C}$ have exactly $n/d$ cells, any maximal independent set in $\mathcal{C}$ has exactly $d$ nodes (districts). It is known that we can compute a $d/2$-approximation for Maximum Weight Independent Set in $d$-Claw free graphs \cite{B2000_ApproximationMaximumWeight}. Therefore we have the following result:

\begin{restatable}{thm}
If all districts have exactly $n/d$ cells, and $d = \Theta(n)$, then we can compute $(n/d +1)/2$-approximation to the swing voters problem in polynomial time.
\end{restatable}

\end{document}